\documentclass[12pt]{article}
\usepackage{latexsym,url}
\usepackage{amsmath,amssymb,array,graphicx,verbatim}
\usepackage{setspace}
\onecolumn
\usepackage
  [breaklinks,bookmarks,bookmarksnumbered,bookmarksopen,bookmarksopenlevel=2]
  {hyperref}
{\makeatletter \hypersetup{pdftitle={\@title}}}

{\makeatletter
 \gdef\xxxmark{%
   \expandafter\ifx\csname @mpargs\endcsname\relax 
     \expandafter\ifx\csname @captype\endcsname\relax 
       \marginpar{xxx}
     \else
       xxx 
     \fi
   \else
     xxx 
   \fi}
 \gdef\xxx{\@ifnextchar[\xxx@lab\xxx@nolab}
 \long\gdef\xxx@lab[#1]#2{{\bf [\xxxmark #2 ---{\sc #1}]}}
 \long\gdef\xxx@nolab#1{{\bf [\xxxmark #1]}}
}

\let\realbfseries=\bfseries
\def\bfseries{\realbfseries\boldmath}

\newif\ifabstract
\abstracttrue
\newif\iffull
\ifabstract \fullfalse \else \fulltrue \fi

\let\epsilon=\varepsilon

\newsavebox{\theorembox}
\newsavebox{\factbox}
\newsavebox{\lemmabox}
\newsavebox{\remarkbox}
\newsavebox{\corollarybox}
\newsavebox{\propositionbox}
\newsavebox{\examplebox}
\newsavebox{\conjecturebox}
\newsavebox{\algbox}
\newsavebox{\qbox}
\newsavebox{\problembox}
\newsavebox{\definitionbox}
\newsavebox{\assumptionbox}
\newsavebox{\hypothesisbox}
\savebox{\theorembox}{\noindent\bf Theorem}
\savebox{\factbox}{\noindent\bf Fact}
\savebox{\lemmabox}{\noindent\bf Lemma}
\savebox{\remarkbox}{\noindent\bf Remark}
\savebox{\corollarybox}{\noindent\bf Corollary}
\savebox{\propositionbox}{\noindent\bf Proposition}
\savebox{\examplebox}{\noindent\bf Example}
\savebox{\conjecturebox}{\noindent\bf Conjecture}
\savebox{\algbox}{\noindent\bf Algorithm}
\savebox{\definitionbox}{\noindent\bf Definition}
\savebox{\problembox}{\noindent\bf Problem}
\savebox{\assumptionbox}{\noindent\bf Assumption}
\savebox{\hypothesisbox}{\noindent\bf Hypothesis}
\newtheorem{theorem}{\usebox{\theorembox}}
\newtheorem{lemma}[theorem]{\usebox{\lemmabox}}

\newtheorem{remark}[theorem]{\usebox{\remarkbox}}

\newcommand{\qed}{\;\;\;\Box}

\newenvironment{proof}{\par{\bf Proof:}}{\(\qed\) \par}
\begin{document}


\title{Games on Social Networks: On a Problem Posed by Goyal}

\author{Ali Kakhbod~\thanks{University of Michigan, Ann Arbor.
                Email: {\tt akakhbod@umich.edu} }
\and Demosthenis~Teneketzis~\thanks{University of Michigan, Ann Arbor.
                Email: {\tt teneket@umich.edu}}
}

\maketitle
\begin{abstract}
Within the context of games on networks S. Goyal (Goyal (2007), pg. 39) posed the following problem. Under any arbitrary but fixed topology, does there exist at least one pure Nash equilibrium that exhibits a positive relation between the cardinality of a player's set of neighbors and its utility payoff? In this paper we present a class of  topologies in which pure Nash equilibria with the above property do not exist. 
\end{abstract}

\begin{section}{Introduction-Motivation}
Games on social networks is a rapidly developing discipline the importance of which has been extensively discussed in recent publications Goyal (2007), Jackson (2008) and Vega-Redondo (2007). One of the goals of this discipline is the development of a framework within which the effect of the structure of relationships on individual behavior and well-being as well as on aggregate outcomes can be examined systematically. \\

The social sharing of information (Bramoulle and Kranton (2007)) which is privately costly to collect is one of the economic contexts within which the effect of network structure on individual behavior and aggregate outcomes are currently being  studied. Within this context Goyal posed a question that addresses the effect of  ones's neighbors on his/her payoff at equilibrium. In this paper we provide an answer to Goyal's problem.\\

The paper is organized as follows: In section \ref{i} we present the social network model and Goyal's problem. In section \ref{ii} we present preliminary results that are essential in determining a solution to Goyal's problem. In section \ref{iii} we present an example that provides a negative answer to Goyal's problem. We conclude in section \ref{iv}.    
\end{section}
\begin{section}{The Model and Goyal's Problem}
\label{i}
We consider a set $\mathcal{N}:=\{1,2,3\cdots,\texttt{N}\}$ of players. These players are connected, directly or indirectly, with one another through a network whose topology is denoted by $\textbf{g}$; $\textbf{g}$ is fixed. Since $\textbf{g}$ is fixed the set $N_i(\textbf{g})$, $i \in \mathcal{N}$, of neighbors of player $i$ is well-defined; two players $i,j$ are called neighbors if they are directly connected with a link in the graph that corresponds to the topology $\textbf{g}$. Each player can choose a strategy/effort $s_i \in \textbf{S}$, where \textbf{S} is compact and convex in $\mathbb{R_{+}}$. Let $s \in \textbf{S}^{\texttt{N}}$, $s:=(s_1,s_2,\cdots,s_{\texttt{N}})$ denote a strategy profile. Given $s \in \textbf{S}$, the payoff of player $i, i \in \mathcal{N},$ is given
\begin{eqnarray}
\Pi_i(s|\textbf{g})=f\left[s_i+\sum_{s_j\in N_i(\textbf{g})}s_j\right]-\textbf{c}s_i,
\end{eqnarray}
where $\textbf{c}$ is the marginal cost of effort and it is a fixed positive constant. The function $f$ is such that $f(0)=0,f'(\cdot)>0$,  $f''(\cdot)<0,$ and there is a constant $\delta$ such that $f'(\delta)=\textbf{c}$. The function $f$ and the constants $\textbf{c}$ and $\delta$ are common knowledge among the players. Consider a game played by the $\texttt{N}$ players, with  strategy space and utility functions described above. Assume for the moment that there exist Nash equilibria in pure strategies,  called pure Nash equilibria (\textit{PNE}), for this game. Denote by $\textit{PNE}_{\textbf{g}}$ the set of $\textit{PNE}$ of the game when the network topology is $\textbf{g}$. Denote by $\textit{PNE}_{\textbf{g}}^+$ the subset of $\textit{PNE}_{\textbf{g}}$ that have the following property:\\
\begin{eqnarray}
\textit{PNE}_{\textbf{g}}^+=\Big\{s^*\in \textit{PNE}_{\textbf{g}}: \forall \ i,j \in \mathcal{N}, |N_i(\textbf{g})|>|N_j(\textbf{g})| \ \Rightarrow \ \Pi_i(s^*|\textbf{g})\geq \Pi_j(s^*|\textbf{g})\Big\},
\end{eqnarray}
where $|N_i(\textbf{g})|$ denotes the cardinality of the set $N_i(\textbf{g})$.
$\textit{PNE}_{\textbf{g}}^+$ is called the set of $\textit{PNE}_{\textbf{g}}$ that exhibit for every player $i \in \mathcal{N}$ a positive relation between the cardinality of $N_i(\textbf{g})$ and the player's utility/payoff under $\textbf{g}$. \\

Let $\mathcal{G}$ denote the set of the all possible topologies of the  $\texttt{N}$ players' network. Goyal (2007),  pg. 39, posed the following problem:\\

\textbf{Goyal's Problem}: \textit{Under any arbitrary but fixed topology $\textbf{g} \in \mathcal{G}$ is the set $\textit{PNE}_{\textbf{g}}^+$ nonempty?}\\

One would expect a positive answer to the above question. In the model under consideration the players' payoff functions are identical. Thus, it is reasonable to expect that players with a higher number of neighbors would earn a higher  payoff at equilibrium. In this paper we provide a negative (thus, in our opinion counterintuitive) answer to Goyal's problem. We present, via an example,  a class of topologies for which $\textit{PNE}_{\textbf{g}}^+$ is empty.
\end{section}
\begin{section}{Preliminaries}
\label{ii}
First we assert that for every topology $\textbf{g}\in \mathcal{G}$ the corresponding game possesses $\textit{PNE}$. This assertion follows from the compactness of the space $\textbf{S}$, continuity of $f(\cdot)$ in the players' strategies, the concavity of $f(\cdot)$ in each player's strategy, and a result of Debreu, Glicksberg and Fan (Fudenberg and Tirole (1991) pg. 34). \\
For any $\textbf{g} \in \mathcal{G}$ all $\textit{PNE}_{\textbf{g}}$ can be characterized by the following result. 
\begin{theorem}[Goyal (2007), pg. 35]
\label{yek}
For a given $\textbf{g}\in \mathcal{G}$, a strategy profile $s^*=(s_1^*,s_2^*,\cdots,s_{\texttt{N}}^*)\in \textit{PNE}_{\textbf{g}}$, if and only if for every $i \in \mathcal{N}$
\begin{eqnarray}
s_i^*= \left\{ \begin{array}{ll}
         \delta-\sum_{j\in N_i(\textbf{g})}s_j^*, & \mbox{if $\sum_{j\in N_i(\textbf{g})}s_j^* < \delta$},\\
        0 & \mbox{if $\sum_{j\in N_i(\textbf{g})}s_j^* \geq \delta$}.\end{array} \right.
\end{eqnarray}
\end{theorem}
The following lemma reveals a property of $\textit{PNE}_{\textbf{g}}^+$. This property along with the result of Theorem \ref{yek} are critical in establishing our answer to Goyal's problem.
\begin{lemma}
\label{lem}
Let $s^*:=(s_1^*,s_2^*,\cdots,s_\texttt{N}^*)\in \textit{PNE}_{\textbf{g}}^+$. Then for any $u,v \in \mathcal{N}$, if $|N_u(\textbf{g})|<|N_v(\textbf{g})|$ and $s_u^*=0$, then $s_v^*=0$.
\end{lemma}
\begin{proof}
We prove the result by contradiction. Consider $s^*\in \textit{PNE}_{\textbf{g}}^+$. Let $u,v \in \mathcal{N}$, assume that $|N_u(\textbf{g})|<|N_v(\textbf{g})|$, $s_u^*=0$, and suppose $s_v^*=\epsilon>0.$ Then, 
since $s^*\in \textit{PNE}_{\textbf{g}}$, by Theorem \ref{yek} we must have 
\begin{eqnarray}
\label{3}
\sum_{j \in N_u(\textbf{g})}s_{j}^* \geq \delta.
\end{eqnarray} 
Furthermore, since $\epsilon>0$, $f$ is increasing, and $\sum_{j \in N_u(\textbf{g})}s_{j}^* \geq \delta$, it follows that
\begin{eqnarray}
\Pi_v(s^*|\textbf{g})=f(\delta)-\textbf{c}\epsilon < f\left(\sum_{j \in N_u(\textbf{g})}s_{j}^*\right)=\Pi_u(s^*|\textbf{g}),
\end{eqnarray}
which contradicts the fact that $s^* \in \textit{PNE}_{\textbf{g}}^+$. Consequently, we must have
\begin{eqnarray}
s_v^*=0.
\end{eqnarray}
\end{proof}
\begin{remark}
Using the arguments in the proof of Lemma \ref{lem} we can also establish the following results: 
\begin{itemize}
	\item Independently of an agent's degree, if $s^* \in \textit{PNE}_{\textbf{g}}$ and $s_i^*<s_j^*$ then $\Pi_i(s^*|\textbf{g})>\Pi_j(s^*|\textbf{g})$. Equivalently, if at equilibrium agent $i$ is no better off than agent $j$ then agent $i$ contributes no more than agent $j$.
  \item Let $s^* \in \textit{PNE}_{\textbf{g}}$ be such that $\sum_{k \in N_i(\textbf{g})}s_{k}^*+s_i^*=\mu$ for all $i\in \mathcal{N}.$ Then, 
  \begin{eqnarray*}
  s_i^* \leq s_j^* \quad \Leftrightarrow\quad \Pi_i(s^*|\textbf{g})\geq \Pi_j(s^*|\textbf{g}).
  \end{eqnarray*}
\end{itemize}
\end{remark}
Using the results on the existence of \textit{PNE}, Theorem \ref{yek} and Lemma \ref{lem}, we present  an example that provides a negative answer to Goyal's problem.
\end{section}
\begin{section}{An Example}
\label{iii}
Consider the network depicted in Figure \ref{fig2}, and denote it by $\textbf{g}$. Assume there exists $s^* \in \textit{PNE}_{\textbf{g}}^+$. Restrict attention to node $v_2$. We show that no matter what value $s_{v_2}^*$  takes, $s^*$ can not be  in $\textit{PNE}_{\textbf{g}}^+$, thus $\textit{PNE}_{\textbf{g}}^+$ is empty. We consider six cases that exhaust all possibilities. 
\begin{itemize}
	\item \textbf{Case 1}: $s_{v_2}^*=0$.\\
	Then $s_{v_2}^*=0$ along with Lemma \ref{lem} imply that 
	\begin{eqnarray}
	\label{6}
	&&s_{v_3}^*=s_{v_4}^*=s_{v_5}^*=0,
	\end{eqnarray}
	therefore,
	\begin{eqnarray}
	&&s_{v_2}^*+s_{v_3}^*+s_{v_4}^*+s_{v_5}^*=0
	\end{eqnarray}
	Since by assumption $s^* \in \textit{PNE}_{\textbf{g}}^+$, $s^* \in\textit{PNE}_{\textbf{g}}$. Therefore, $s^*$ must satisfy the condition of Theorem \ref{yek}. This further implies that 
\begin{eqnarray}
\label{7}
s_{v_3}^*+\sum_{v_j \in N_{v_3}(\textbf{g})}s_{v_j}^*=s_{v_3}^*+s_{v_2}^*+s_{v_4}^*+s_{v_5}^*\geq \delta.
\end{eqnarray}
But \eqref{7} contradicts \eqref{6}, thus we can not have any $s^* \in \textit{PNE}_{\textbf{g}}^+$ with $s_{v_2}^*=0.$
\item \textbf{Case 2}: $0 <s_{v_2}^* < \delta$ and $\Delta:=s_{v_4}^*+s_{v_5}^*=0$.\\
Consider node $v_6$. Since $|N_{v_6}(\textbf{g})|=5>4=|N_{v_5}(\textbf{g})|$ and $s_{v_5}^*=0$, Lemma \ref{lem} implies that 
\begin{eqnarray}
\label{9}
s_{v_6}^*=0. 
\end{eqnarray}
By the same argument, $s_{v_6}^*=0$ implies that 
\begin{eqnarray}
\label{10}
&&s_{v_{16}}^*=s_{v_{17}}^*=s_{v_{18}}^*=s_{v_{19}}^*=0,
\end{eqnarray}
On the other hand, by Theorem \ref{yek} we must have
\begin{eqnarray}
\label{11}
s_{v_6}^*+\sum_{v_j \in N_{v_6}(\textbf{g})}s_{v_j}^*\geq \delta
\end{eqnarray}
Equation \eqref{11}  contradicts Eq. \eqref{10}. 
\item \textbf{Case 3}: $0 <s_{v_2}^* < \delta$, $0<\Delta:=s_{v_4}^*+s_{v_5}^*<\delta$ and $0<\Delta+s_{v_2}^*<\delta$.\\
Since $s^* \in \textit{PNE}_{\textbf{g}}^+$, by Theorem \ref{yek} we must have 
\begin{eqnarray}
\label{12}
s_{v_3}^*+s_{v_4}^*+s_{v_5}^*+s_{v_2}^*=s_{v_3}^*+\Delta+s_{v_2}^*\geq \delta.
\end{eqnarray}
Furthermore, by assumption 
\begin{eqnarray}
\label{13}
\Delta+s_{v_2}^*<\delta.
\end{eqnarray}
Then Eqs. \eqref{12}, \eqref{13} and Theorem \ref{yek} imply that $s_{v_3}^*=\delta-(\Delta+s_{v_2}^*),$ or,  equivalently,  
\begin{eqnarray}
\label{14} 
s_{v_2}^*+s_{v_3}^*=\delta-\Delta<\delta,
\end{eqnarray}
because $\Delta>0.$\\
On other other hand, Theorem \ref{yek} applied to node $v_2$ gives
\begin{eqnarray}
\label{15}
s_{v_2}^*+s_{v_3}^*\geq \delta.
\end{eqnarray}
Equations \eqref{14} and \eqref{15} contradict one another.
\item \textbf{Case 4}: $0 <s_{v_2}^* < \delta$, $0<\Delta:=s_{v_4}^*+s_{v_5}^*<\delta$ and $\Delta+s_{v_2}^* \geq \delta$.\\
In this case,  Theorem \ref{yek} applied to node $v_3$ gives
\begin{eqnarray}
\label{16}
s_{v_3}^*+\sum_{v_j \in N_{v_3}(\textbf{g})}s_{v_j}^*\geq \delta
\end{eqnarray}
Inequality \eqref{16} combined with the assumption $\Delta+s_{v_2}^*\geq \delta$ result in
\begin{eqnarray}
\label{17}
s_{v_3}^*=0.
\end{eqnarray}
Equation \eqref{17} and Lemma \ref{lem} in turn imply that 
\begin{eqnarray}
\label{18}
s_{v_5}^*=s_{v_6}^*=s_{v_7}^*=s_{v_8}^*=0.
\end{eqnarray}
On the other hand, by Theorem \ref{yek} we must have
\begin{eqnarray}
\label{20}
s_{v_5}^*+\sum_{v_j \in N_{v_5}(\textbf{g})}s_{v_j}\geq \delta.
\end{eqnarray}
which contradicts \eqref{18}.
\item \textbf{Case 5}: $0 <s_{v_2}^* < \delta$, $\Delta:=s_{v_4}^*+s_{v_5}^* \geq \delta$.\\
The result of Theorem  \ref{yek} for node $v_3$ along with the above conditions imply 
\begin{eqnarray}
\label{21}
s_{v_3}^*=0.
\end{eqnarray}
Then, by Eq. \eqref{21} and Lemma \ref{lem} we must have 
\begin{eqnarray}
\label{22}
s_{v_4}^*=s_{v_5}^*=0,
\end{eqnarray}
%
which contradicts the assumption 
\begin{eqnarray}
\label{24}
s_{v_4}^*+s_{v_5}^*\geq \delta.
\end{eqnarray} 
\item \textbf{Case 6}: $s_{v_2}^*\geq \delta$.\\
Then, the result of  Theorem \ref{yek} for node $v_2$ and $s_{v_2}^*\geq \delta$ imply  
\begin{eqnarray}
\label{26}
s_{v_3}^*=0.
\end{eqnarray}
Equation \eqref{26} and lemma \ref{lem} result in
\begin{eqnarray}
\label{226}
s_{v_5}^*=0,
\end{eqnarray}
which along with Lemma \ref{lem} imply that
\begin{eqnarray}
\label{27}
s_{v_6}^*=s_{v_7}^*=s_{v_8}^*=0.
\end{eqnarray}
On the other hand, the result of Theorem \ref{yek} for node $v_5$ requires that 
\begin{eqnarray}
\label{29}
s_{v_5}^*+\sum_{v_j \in N_{v_5}(\textbf{g})}s_{v_j}\geq \delta
\end{eqnarray}
which contradicts \eqref{27}. 
\end{itemize}
The analysis of cases 1-6 demonstrates that no matter what value $s_{v_2}^*$ takes, $s^* \notin \textit{PNE}_{\textbf{g}}^+$. Consequently, the set
$\textit{PNE}_{\textbf{g}}^+$ is empty.
 \begin{figure}[t]
\label{fig2}
\centering
\includegraphics[height=5in,width=5in]{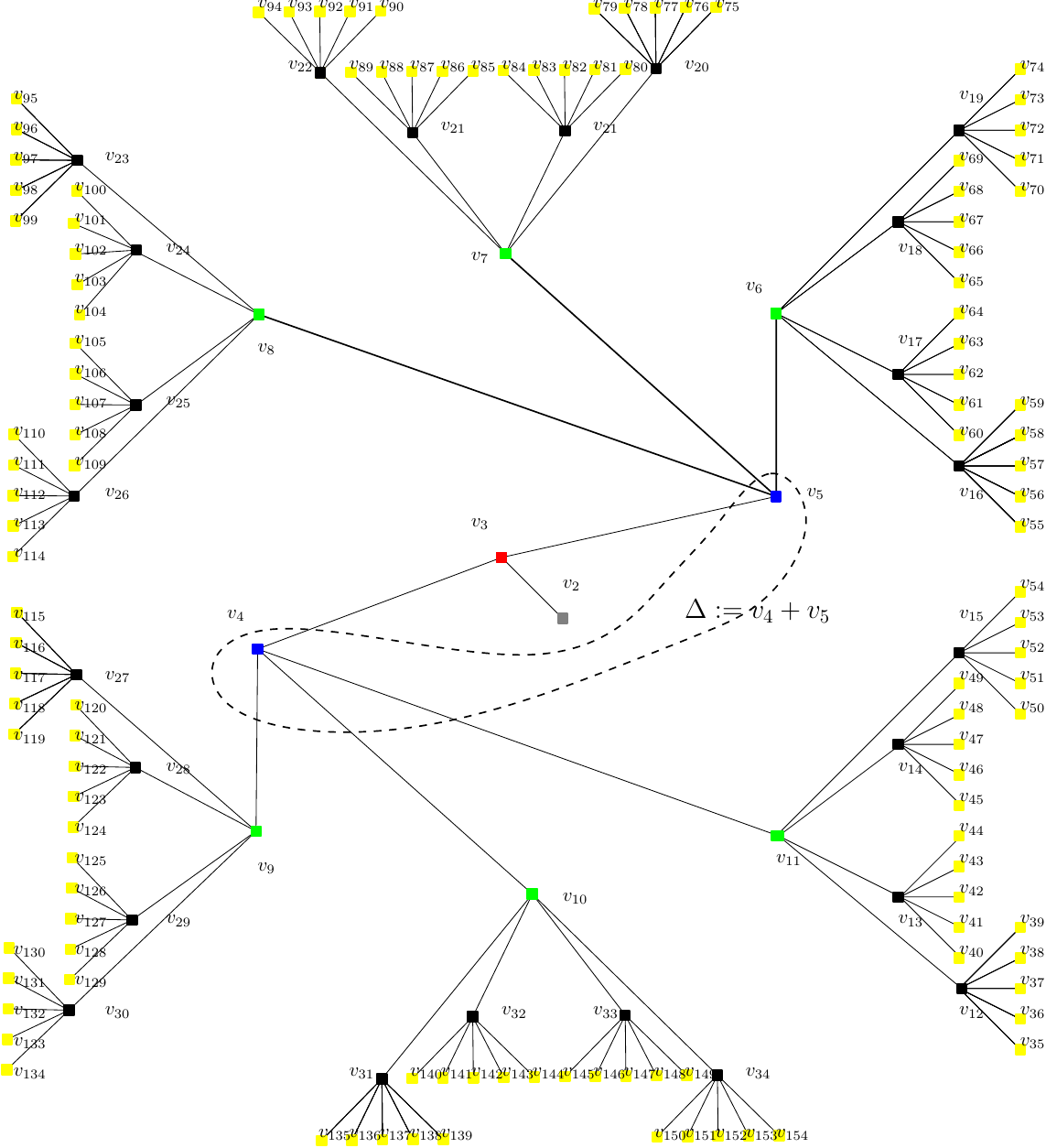}
\caption{$\textbf{g}$: An example}
\end{figure}  
\end{section}
\begin{section}{Conclusion/Discussion}
\label{iv}
Note that the example presented here is not necessarily minimal, but it introduces \emph{a class of topologies (games)} with  finite \textbf{tree structure}  such that\\
\begin{itemize}
\item (i) \textit{the degree of nodes (except the terminal/leaf nodes) increases monotonically as their distance from the root of the tree increases. }
\item (ii) \textit{the depth of the tree is at least 5.} 
\end{itemize}
For any game on a network with topology $\textbf{g}$ that belongs to the above described class, the set $\textit{PNE}_{\textbf{g}}^+$ is empty. This can be proved by \emph{exactly} the same arguments (Cases 1-6) as those presented in the tree structure of Section \ref{iii}. \\

Furthermore, consider any topology/game that results in from \emph{arbitrary} extensions along the leaf nodes of the tree structure of Section \ref{iii}. By the same arguments as in Section \ref{iii} we can show that $\textit{PNE}_{\textbf{g}}^+$ is empty for those topologies/games.\\

The form of the utility function, namely the fact that the function $f$ has all the properties stated in the model and is the same for all players, is critical in establishing a negative answer to Goyal's problem. The class of network   topologies $\textbf{g}$ for which $\textit{PNE}_{\textbf{g}}^+$ is nonempty remains unknown. 
\end{section}
\\

\noindent{\Large{\textbf{References}}}\\
\noindent{Bramoulle, Y. and R.~Kranton (2007), {''Public Goods in Networks''}, \textit{Journal of Economic Theory} 135, 478-494.}\\

\noindent{Goyal, S. (2007)  {''Connections, An Introduction to the economics of Networks''}, Princeton University Press.}\\

\noindent{Jackson, M.  (2008) {''Social and Economic Networks''}, Princeton University Press.}\\

\noindent{Vega-Redondo, F. (2007) {''Complex Social Networks''}, Cambridge University Press.}\\

\noindent{Fudenberg, D.  and Tirole, J. (1991) {''Game Theory''}, MIT Press.}


\end{document}